\documentclass[a4paper,UKenglish,cleveref,thm-restate]{lipics-v2021}
\usepackage{graphicx,amssymb,amsmath}

\usepackage{complexity}
\usepackage{nicefrac}
\usepackage{mathtools}
\usepackage{amsfonts}
\usepackage{thmtools}
\usepackage{xspace}
\usepackage[low-sup]{subdepth}
\newcommand{\BILLE}{\makebox{\textsc{Bill-E}}\xspace}
\newcommand{\ARMADAS}{\makebox{\textsc{Armadas}}\xspace}
\newcommand{\probName}{{\textsc{Single Robot Reconfiguration}}\xspace}

\newcommand{\BigO}{\mathcal{O}}
\newcommand{\OPT}{\textsf{OPT}\xspace}

\newcommand{\configs}{\ensuremath{\mathcal{C}}}
\newcommand{\bipGraph}[1]{G_{#1}}
\newcommand{\sumcomps}[1]{F_\Sigma}

\DeclarePairedDelimiter{\norm}{\lvert}{\rvert}
\newcommand{\eDist}[1]{d_E(#1)}
\newcommand{\cDist}[1]{d_C(#1)}
\newcommand{\weight}[1]{w(#1)}

\newcommand{\minMatching}{M_{\OPT}}

\newcommand{\lowerboundOf}[1]{\ensuremath{\sigma(#1)}}

\usepackage[subrefformat=simple,labelformat=simple]{subcaption}

\graphicspath{{./figures/}}

\hideLIPIcs
\nolinenumbers

\title{Efficient Reconfiguration of Tile Arrangements by a Single Active Robot}
\titlerunning{Efficient Reconfiguration of Tile Arrangements by a Single Active Robot}

\author{Aaron T. Becker}{Electrical Engineering, University of Houston, Texas, USA}{atbecker@uh.edu}{https://orcid.org/0000-0001-7614-6282}{}
\author{Sándor P. Fekete}{Computer Science, TU Braunschweig, Germany}{s.fekete@tu-bs.de}{https://orcid.org/0000-0002-9062-4241}{}
\author{Jonas Friemel}{Electrical Engineering and Computer Science, Bochum University of Applied Sciences, Germany}{jonas.friemel@hs-bochum.de}{https://orcid.org/0009-0009-6270-4779}{}
\author{Ramin Kosfeld}{Computer Science, TU Braunschweig, Germany}{kosfeld@ibr.cs.tu-bs.de}{https://orcid.org/0000-0002-1081-2454}{}
\author{Peter Kramer}{Computer Science, TU Braunschweig, Germany}{kramer@ibr.cs.tu-bs.de}{https://orcid.org/0000-0001-9635-5890}{}
\author{Harm Kube}{Computer Science, TU Berlin, Germany}{h.kube@campus.tu-berlin.de}{https://orcid.org/0009-0001-5072-7908}{}
\author{Christian Rieck}{Discrete Mathematics, University of Kassel, Germany}{christian.rieck@mathematik.uni-kassel.de}{https://orcid.org/0000-0003-0846-5163}{}
\author{Christian Scheffer}{Electrical Engineering and Computer Science, Bochum University of Applied Sciences, Germany}{christian.scheffer@hs-bochum.de}{https://orcid.org/0000-0002-3471-2706}{}
\author{Arne Schmidt}{Computer Science, TU Braunschweig, Germany}{aschmidt@ibr.cs.tu-bs.de}{https://orcid.org/0000-0001-8950-3963}{}

\authorrunning{Becker, Fekete, Friemel, Kosfeld, Kramer, Kube, Rieck, Scheffer, Schmidt}

\Copyright{Aaron T. Becker, Sándor P. Fekete, Jonas Friemel, Ramin Kosfeld, Peter Kramer, Harm Kube, Christian Rieck, Christian Scheffer, Arne Schmidt}



\funding{
    This work was partially supported by multiple sources:
    German Research Foundation~(DFG) project ``Space~Ants'' (FE~407/22-1 and SCHE 1931/4-1) and grant~522790373, as well as by grants NSF~IIS-2130793, ARL~W911NF-23-2-0014, and DAF~AFX23D-TCSO1.
	}

\begin{document}
    \maketitle
    \begin{abstract}
        We consider the problem of reconfiguring a two-dimensional connected grid arrangement of passive building blocks from a start configuration to a goal configuration, using a single active robot that can move on the tiles, remove individual tiles from a given location and physically move them to a new position by walking on the remaining configuration.
        The objective is to determine a schedule that minimizes the overall makespan, while keeping the tile configuration connected.
        We provide both negative and positive results.
        (1) We generalize the problem by introducing weighted movement costs, which can vary depending on whether tiles are carried or not, and prove that this variant is \NP-hard.
        (2) We give a polynomial-time constant-factor approximation algorithm for the case of disjoint start and target bounding boxes, which additionally yields optimal carry distance for 2-scaled instances.
    \end{abstract}

    \section{Introduction}\label{sec:introduction}
Building and modifying structures consisting of many basic components is an
important objective, both in fundamental theory and in a spectrum of practical settings.
Transforming such structures with the help of autonomous robots is
particularly relevant in very large~\cite{2020-Coordinating_IWOCA} and
very small dimensions~\cite{Yangsheng:five_mm_robot}
that are hard to access for direct human manipulation, e.g.,
in extraterrestrial space~\cite{BENLARBI20213598,jenett2017design} or in microscopic environments~\cite{bfh+-cgur-17}.
This gives rise to the natural algorithmic problem of rearranging a
given start configuration of many \emph{passive} objects
by a small set of \emph{active} agents to a desired target configuration.
Performing such reconfiguration at scale faces a number of critical challenges,
including (i)~the cost and weight of materials, (ii)~the potentially disastrous
accumulation of errors, (iii)~the development of simple yet resilient agents to carry
out the active role, and (iv)~achieving overall feasibility and~efficiency.

In recent years, significant advances have been made to deal with these difficulties.
Macroscopically, ultra-light and scalable composite lattice materials~\cite{Cheung1219,ultralight24,gregg2018ultra} tackle the first problem, making it possible to construct modular, reconfigurable structures with platforms such as NASA's \BILLE and \ARMADAS~\cite{ultralight24,jenett2017bille,jenett2016meso};
the underlying self-adjusting lattice also resolves the issues of accuracy and error correction (ii),
allowing it to focus on discrete, combinatorial structures,
consisting of regular tiles (in two dimensions) or voxels (in three dimensions).
A further step has been the development of simple autonomous
robots~\cite{9836082,jenett2019material}
that can be used to carry out complex tasks (iii),
as shown in~\cref{fig:intro-bille}: The robot can move on the tile arrangement, remove
individual tiles and physically relocate them to a new position by walking on
the remaining configuration, which needs to remain connected at all times.
At microscopic scales, advances in micro- and nanobots~\cite{santos2022light,zhang2019robotic} offer novel ways to (re)configure objects and mechanisms, e.g., assembling specific structures or gathering in designated locations for tasks like targeted drug delivery~\cite{2020-Targeted_ICRA,2022-gather_IROS}.

\begin{figure}
    \centering
    \includegraphics[scale=0.23]{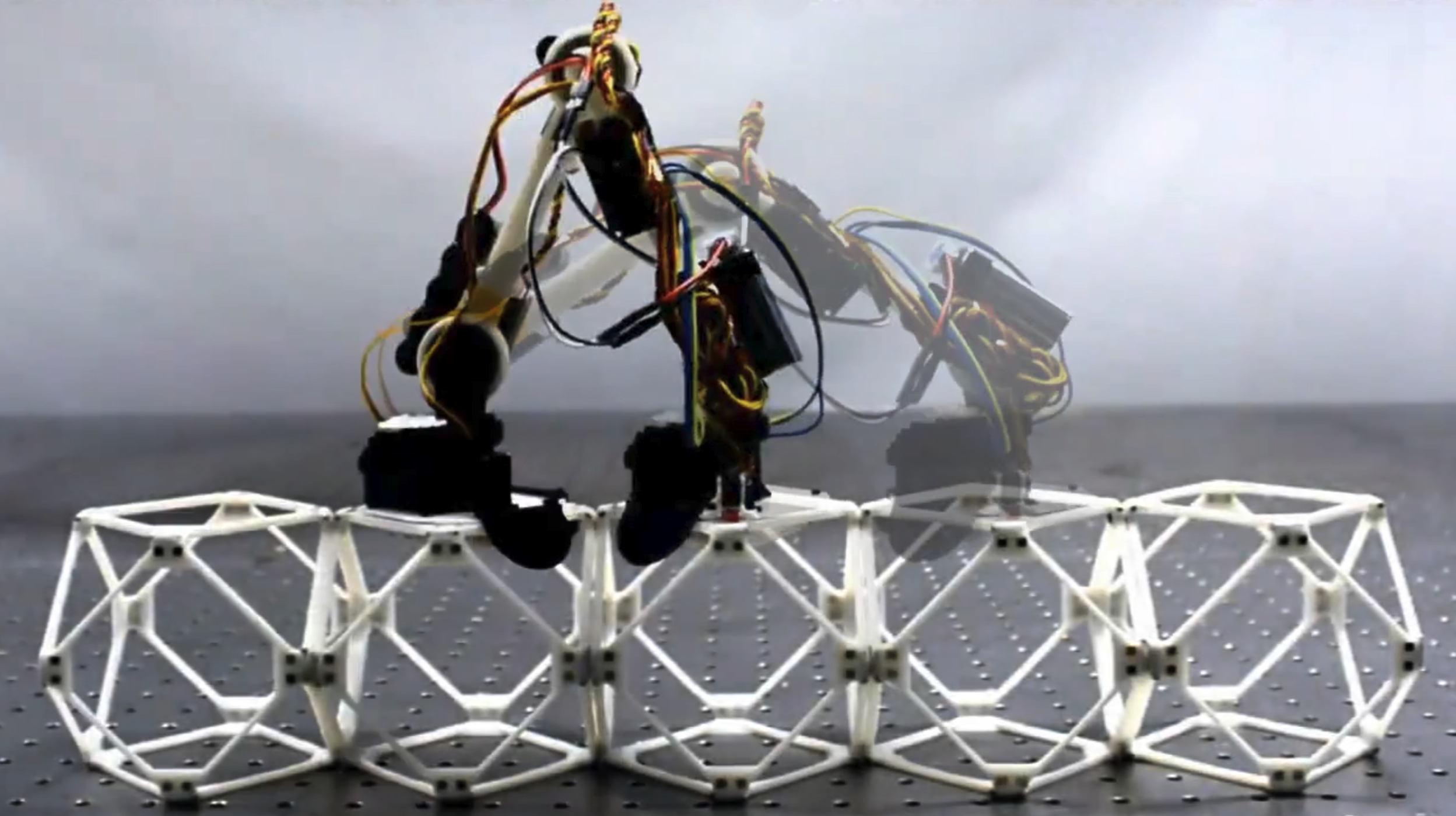}
    \caption{A \BILLE robot moves on a configuration of light-weight
material and relocate individual voxels for overall reconfiguration. Photo adapted from~\cite{jenett2019material}.}
    \label{fig:intro-bille}
\end{figure}

In this paper, we deal with challenge (iv):
How can we use such a robot to transform a given start configuration into a desired goal arrangement, as quickly as possible?

\subsection{Our contributions}\label{subsec:our-results}

We investigate the problem of finding minimum cost reconfiguration schedules for a single active robot operating on a (potentially large) number of tiles, and give the following results.
\begin{enumerate}[(1)]
	\item {We present a generalized version of the problem, parameterized by movement costs, which can vary depending on whether tiles are carried or not,
	and show that this is \NP-hard.}
	\item {We give a polynomial-time constant-factor approximation algorithm for the case of disjoint start and target bounding boxes.
	Our~approach further yields optimal carry distance for $2$-scaled instances.}
\end{enumerate}

\subsection{Related work}
\label{subsec:related-work}

Recently, the authors of~\cite{single-bille-reconfig-IROS,cooperative-bille-reconfig-ICRA} showed that computing optimal schedules for \BILLE bots, see~\cite{jenett2017bille} and \cref{fig:intro-bille}, is \NP-hard in unweighted models.
They designed a heuristic approach that exploits rapidly exploring random trees~(RRT) and a time-dependent variant of the A$^{*}$ algorithm, as well as target swapping heuristics to reduce the overall distance traveled for multiple robots.
The authors of~\cite{cheung2025assembly} present an algorithm for computing feasible build orders that adhere to three-dimensional tile placement constraints in the \ARMADAS project.

A different context for reconfiguration arises from
programmable matter~\cite{daymude2019computing}.
Here, finite automata are capable of building bounding boxes of tiles around polyominoes, as well as scale and rotate them while maintaining connectivity at all times~\cite{fekete2022connected,NiesReconfig}.
On hexagonal grids, finite automata can build and recognize simple shapes such as triangles or parallelograms~\cite{gmyr2018recognition,gmyr2020forming,hinnenthal2024efficient} as well as more complicated shapes if they are able to recognize nodes that belong to the target shape~\cite{friemel2025reconfiguration}.

When considering active matter, arrangements composed of self-moving objects (or agents), numerous models exist~\cite{almethen2020pushing,almethen2022efficient,connor2025transformation,FeketeKKRS23-journal-connected,michail2019transformation}.
For example, in the \emph{sliding cube model} (or the \emph{sliding square model} in two dimensions) first introduced in 2003~\cite{FitchBR03,FitchBR05}, agents are allowed to slide along other, temporarily static, agents for reconfiguration, but must maintain connectivity throughout.
The authors of~\cite{AkitayaDKKPSSUW22} show that sequential reconfiguration in two dimensions is always possible, but deciding the minimal makespan is \NP-complete.
Recent work presents similar results for the three-dimensional setting~\cite{AbelAKKS24,KostitsynaOPPSS24}, as well as parallel movement~\cite{parallel-sliding-squares}.

In a relaxed model, the authors of~\cite{FeketeKKRS23-journal-connected,FeketeKRS022-journal-labeled-connected} show that parallel connected reconfiguration of swarms of (labeled) agents is \NP-complete, even for makespan~$2$, and present algorithms for schedules with constant \emph{stretch}; the ratio of a schedule's makespan to the individual maximum minimum distance between start and target.

\subsection{Preliminaries}
\label{subsec:preliminaries}

For the following, we refer to~\cref{fig:preliminaries}.
We are given a fixed set of~$n$ indistinguishable square \emph{tile} modules located at discrete, unique positions $(x,y)$ in the infinite integer grid~$\mathbb{Z}^2$.
If their positions induce a connected subgraph of the grid, where two positions are connected if either their $x$- or $y$-coordinate differs by $1$, we say that the tiles form a connected \emph{configuration} or \emph{polyomino}.
Let~$\configs(n)$ refer to the set of all polyominoes of $n$ tiles.

\begin{figure}[htb]
	\begin{subfigure}[t]{\columnwidth}
		\centering%
		\includegraphics[page=1]{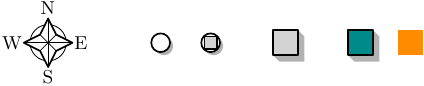}%
		\caption{Left to right: The cardinal directions, a robot without and with payload, a tile, and tiles that only exist in either the start or the target configuration, respectively.}
		\label{fig:legend}
	\end{subfigure}%
	\vspace{1em}
	\begin{subfigure}[t]{\columnwidth}
		\centering%
		\includegraphics[page=2]{intro-moves-colors}%
		\caption{A schedule for the instance on the left. The robot moves on tiles, picking up and dropping off adjacent~tiles.}
		\label{fig:example-schedule}
	\end{subfigure}%
	\caption{A brief overview of legal operations.}
	\label{fig:preliminaries}
\end{figure}

Consider a robot that occupies a single tile at any given time.
In discrete time steps, the robot can either move to an adjacent tile, pick up a tile from an adjacent position (if it is not already carrying one), or place a carried tile at an adjacent unoccupied position.
Tiles can be picked up only if connectivity is preserved.

We use cardinal directions; the unit vectors $(1,0)$ and $(0,1)$ correspond to \emph{east} and \emph{north}, respectively.
Naturally, their opposing vectors extend \emph{west} and \emph{south}.

A \emph{schedule} $S$ is a finite, connectivity-preserving sequence of operations to be performed by the robot.
As the robot's motion is restricted to movements on the polyomino, we refer to distances it travels as \emph{geodesic} distances.
Let~$\cDist{S}$ denote the \emph{carry distance}, which is the number of robot moves while carrying a tile (i.e., the sum of geodesic distances between consecutive pickups and drop-offs), plus the number of all pickups and drop-offs in $S$.
Accounting for the remaining moves without carrying a tile,
the \emph{empty distance} $\eDist{S}$ is the sum of geodesic distances between drop-offs and pickups in~$S$.
For $C_s, C_t \in \configs(n)$, we~say that~$S$ is a schedule for ${C_s\rightrightarrows C_t}$ exactly if it transforms $C_s$ into~$C_t$.

\subparagraph*{Problem statement.}
We consider \probName:
Given configurations ${C_s,C_t\in\configs(n)}$ and a rational weight factor ${\lambda\in[0,1]}$,
determine a schedule $S$ for $C_s\rightrightarrows C_t$ that minimizes the \emph{weighted makespan}
$\norm{S} \coloneqq \lambda\cdot\eDist{S} + \cDist{S}$.
We refer to the minimum weighted makespan for a given instance as~\OPT.

    \section{Computational complexity of the problem}
\label{sec:computational-complexity}
We investigate the computational complexity of the decision variant of the reconfiguration problem.
In particular, we prove that the problem is \NP-hard for any rational $\lambda\in[0,1]$.
This generalizes the previous result from~\cite{cooperative-bille-reconfig-ICRA}, which handles the case of~$\lambda=1$.

\begin{restatable}{theorem}{hardness}
	\label{thm:weighted-hardness}
	\probName is \NP-hard for any rational $\lambda$.
\end{restatable}

We distinguish between the two cases of (1) $\lambda \in (0,1]$ and (2) $\lambda = 0$.
For (1), we modify the construction from~\cite{cooperative-bille-reconfig-ICRA} to yield a slightly stronger statement, reducing from the \textsc{Hamiltonian path} problem in grid graphs.
The primary result is the following.
\begin{lemma}
	\label{lem:hardness-lambda-not-zero}
	\probName is \NP-hard for $\lambda\in (0,1]$.
\end{lemma}
\begin{proof}
	We reduce from the \textsc{Hamiltonian Path} problem in induced subgraphs of the infinite grid graph~\cite{ItaiPS82}, which asks us to decide whether there exists a path in a given graph that visits each vertex exactly once.
	We use a modified variant of the construction in~\cite{cooperative-bille-reconfig-ICRA}, defining our start and target configurations based on the gadgets from~\cref{fig:hardness-bridges} as follows.

	\begin{figure}[htb]
		\subcaptionsetup{justification=centering}%
		\begin{subfigure}[t]{0.3\columnwidth}
			\centering%
			\includegraphics[page=2]{hardness-bridges}%
			\caption{The vertex gadget.}%
			\label{fig:hardness-bridges-vertex}%
		\end{subfigure}%
		\begin{subfigure}[t]{0.7\columnwidth}
			\centering%
			\includegraphics[page=1]{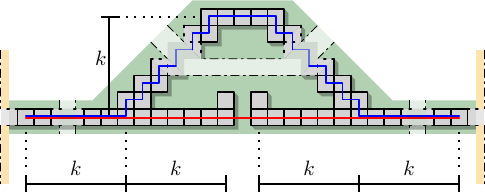}%
			\caption{An edge gadget; note the blue and red paths.}%
			\label{fig:hardness-bridges-edge}%
		\end{subfigure}%
		\caption{Our construction forces the robot to visit each vertex gadget at least once, traversing edge gadgets along either the red or the blue path.}%
		\label{fig:hardness-bridges}
	\end{figure}
	Let $k\gg \nicefrac{1}{\lambda}$.
	For each vertex of a given grid graph~$G_{\mathit{in}}$, we create a copy of a $7\times 7$ \emph{vertex gadget} that represents a constant-size, locally solvable task, see~\cref{fig:hardness-bridges-vertex}.
	We connect all adjacent vertex gadgets with copies of a $(k+1)\times (4k+3)$ \emph{edge gadget}, as depicted in~\cref{fig:hardness-bridges-edge}.
	Unlike~\cite{cooperative-bille-reconfig-ICRA}, our edge gadget offers two feasible paths, see~\cref{fig:hardness-bridges-edge}.

	The blue path offers a simple, walkable option of $6k+2$ empty travel units, while the red path is only available if the robot performs $2$ carry steps to construct and deconstruct a unit length bridge at the center, temporarily creating a walkable path of $4k+2$ units.
	This is cheaper exactly if $\lambda(6k+2)>\lambda(4k+2)+2$,
	i.e., if $k > \nicefrac{1}{\lambda}$.
	Since $\lambda$ is a rational number, we can pick a value $k$ that is polynomial in $\lambda$ and permits a construction in which taking the red path is always cheaper.
	The blue path thus serves solely for the purpose of providing connectivity and will never be taken by the robot.

	An optimal schedule will therefore select a minimal number of edge gadgets to cross, always taking the red path.
	As vertex gadgets always require the same number of moves to solve, we simply capture the total cost for a single vertex gadget as $t_g$.
	For a graph~$G_{\mathit{in}}$ with~$m$ vertices, an optimal schedule has makespan
	$m\cdot t_g + (m-1)\cdot (4k\lambda+2\lambda+2)$
	exactly if the underlying graph is Hamiltonian, and greater otherwise.

	Given a connected grid graph, such a schedule never builds bridges of length greater than one:
	Due to their pairwise distance, this will never be cheaper than crossing an edge gadget along the red path, see also~\cite{cooperative-bille-reconfig-ICRA}.
\end{proof}

As our reduction forces the robot to perform a pair of pickup and drop-off operations to efficiently travel between adjacent vertex gadgets, it immediately follows that deciding the necessary number of these operations in an optimal schedule is \NP-hard.

\begin{corollary}
	Given two configurations $C_s,C_t\in\configs(n)$ and an integer $k\in\mathbb{N}$, it is \NP-hard to decide whether there exists an optimal schedule $C_s\rightrightarrows C_t$ with at most~$k$ pickup (at most $k$ drop-off) operations, if ${\lambda \in (0,1]}$.
\end{corollary}

The reduction for (2) is significantly more involved: As $\lambda = 0$, the robot is effectively allowed to ``teleport'' across the configuration, albeit only without cargo.

\begin{lemma}
	\label{lem:hardness-lambda-zero}
	\probName is \NP-hard for $\lambda = 0$.
\end{lemma}
We reduce from \textsc{Planar Monotone~3Sat}~\cite{dbk-obspp-10}, following ideas by the authors of~\cite{AkitayaDKKPSSUW22} for the sequential sliding square problem.
This variant of the SAT problem asks whether a given Boolean formula $\varphi$ in conjunctive normal form is satisfiable, given the following properties:
First, each clause consists of at most~$3$~literals, all either positive or negative.
Second, the clause-variable incidence graph $G_\varphi$ admits a plane drawing in which variables are placed on the $x$-axis, and all positive (resp., negative) clauses are located in the same half-plane, such that edges do not cross the $x$-axis, see~\cref{fig:hardness-teleportation-graph}.
We~assume, without loss of generality, that each clause contains exactly three literals; otherwise, we extend a shorter clause with a redundant copy of one of the existing literals, e.g., $(x_1\lor x_4)$ becomes $(x_1\lor x_4\lor x_4)$.

\begin{figure}[htb]
	\centering%
	\includegraphics[page=6]{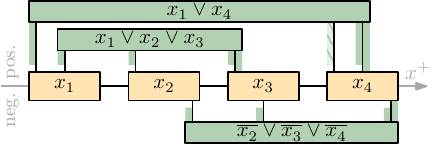}%
	\caption{The embedded clause-variable incidence graph of an instance of \textsc{Planar~Monotone~3Sat}, i.e., the Boolean formula ${\varphi= (x_1\lor x_3\lor x_4)\land (x_1\lor x_4)\land (\overline{x_2}\lor\overline{x_3}\lor\overline{x_4})}$.}%
	\label{fig:hardness-teleportation-graph}%
\end{figure}

Our reduction maps from an instance $\varphi$ of \textsc{Planar Monotone 3Sat} to an instance~$\mathcal{I}_\varphi$ of \probName such that the minimal feasible makespan~$\mathcal{I}_\varphi$ is determined by whether $\varphi$ is satisfiable.
Recall that, due to $\lambda=0$, we only account for carry distance.
Consider an embedding of the clause-variable incidence graph as above, where $C$ and $V$ refer to the $m$ clauses and $n$ variables of $\varphi$, respectively.

We construct $\mathcal{I}_\varphi$ as follows.
A \emph{variable~gadget} is placed on the $x$-axis for each ${x_i\in V}$, and connected along the axis in a straight line.
Intuitively, the variable gadget asks the robot to move a specific tile west along one of two feasible paths, which encode a value assignment.
These paths are highlighted in red and blue in~\cref{fig:hardness-teleportation-variable}, each of length exactly $9(\delta(x_i)+1)$, where $\delta(x_i)$ refers to the degree of $x_i$ in the clause-variable incidence graph.
\begin{figure}[htb]
	\centering%
	\includegraphics[page=2]{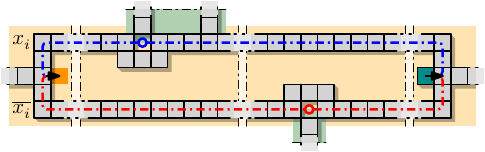}%
	\caption{The variable gadget; the middle segments are repeated to match the incident clauses.}%
	\label{fig:hardness-teleportation-variable}%
\end{figure}
Both paths require the robot to place its payload into the highlighted gaps to minimize the makespan, stepping over it before picking up again.
We further add one \emph{clause gadget} per clause, connected to its incident variables.
These are effectively combs with three or four prongs, connected to a variable gadget by the rightmost one, see~\cref{fig:hardness-teleportation-clause}.
\begin{figure}[h!tb]%
	\centering%
	\includegraphics[page=3]{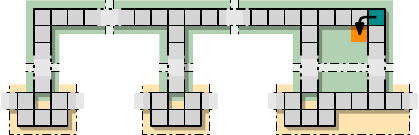}%
	\caption{A clause gadget and its variable attachment.}%
	\label{fig:hardness-teleportation-clause}%
\end{figure}%
The remaining prongs each terminate at a position diagonally adjacent to the variable gadget of a corresponding literal.
To solve the gadget, the rightmost prong must be temporarily disconnected from the comb, meaning that another prong must have been connected to a variable gadget first.
Such a link can be established without extra cost if one of the adjacent variable gadgets has a matching Boolean value, allowing us to use its payload to temporarily close one of the gaps marked by circles in~\cref{fig:hardness-teleportation-variable}.

We can show that a weighted makespan of $29m+9n$ can be achieved exactly if $\varphi$ is satisfiable; recall that $m$ and $n$ denote the number of clauses and variables in~$\varphi$, respectively.

\begin{proof}[Proof of~\cref{lem:hardness-lambda-zero}]
	Consider an instance $\varphi$ of \textsc{Planar Monotone 3Sat} with $n$ variables and $m$ clauses.
	Recall that each clause can be assumed to contain exactly three literals, i.e, $\sum_{i=1}^{n}\delta(x_i) = 3m$.
	We show that $\varphi$ is satisfiable exactly if there exists a schedule for the constructed instance $\mathcal{I}_\varphi = (C_\varphi, C_\varphi')$ of the reconfiguration problem with weighted makespan
	\begin{equation}
		\label{eq:sat-hardness-makespan}
		2m + \sum_{i=1}^{n}9(\delta(x_i)+1) = 2m + 9(3m+n) = 29m+9n.
	\end{equation}
	\begin{claim}
		\label{clm:hardness-teleportation-positive}
		There exists a schedule of weighted makespan $29m+9n$ for $\mathcal{I}_\varphi$ if $\varphi$ is satisfiable.
	\end{claim}
	\begin{claimproof}
		Consider a satisfying assignment $\alpha(\varphi)$ for $\varphi$ and let~${\alpha(x_i) \in\{{x_i},\overline{x_i}\}}$ represent the literal of $x_i$ with a positive Boolean value, i.e., the literal such that $\alpha(x_i)=\texttt{true}$.
		Using only~$\alpha(\varphi)$ and the paths illustrated in~\cref{fig:hardness-teleportation-variable}, we can construct a schedule for $\mathcal{I}_\varphi$ as follows.

		We follow the path corresponding to~$\alpha(x_i)$, picking up the tile and placing it at the first unoccupied position.
		If this position is adjacent to a clause gadget's prong, this corresponds to the pickup/drop-off sequence $(a)$ in \cref{fig:hardness-teleportation-variable-example}.
		\begin{figure}[htb]
			\centering%
			\includegraphics[page=4]{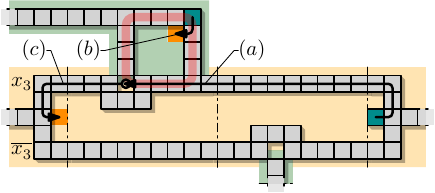}%
			\caption{The variable gadget for $x_3$ as seen in~\cref{fig:hardness-teleportation-graph} and a schedule with makespan ${29=9(\delta(x_3)+1) + 2}$ that solves an adjacent clause gadget.}%
			\label{fig:hardness-teleportation-variable-example}%
		\end{figure}
		With the marked cell occupied, there then exists a cycle containing the tile that is to be moved in the clause gadget, allowing us to solve the clause gadget in just two steps, marked $(b)$ in \cref{fig:hardness-teleportation-variable-example}.
		We can repeat this process, moving the tile to the next unoccupied position on its path until it reaches its destination.

		This pattern takes makespan exactly $9(\delta(x_i)+1)$ per variable gadget; $9$ for each incident clause and $9$ as the base cost of traveling north and south.

		As $\alpha(\varphi)$ is a satisfying assignment for $\varphi$, every clause of $\varphi$ contains at least one literal $v$ such that $\alpha(v)=\texttt{true}$.
		It follows that every clause gadget can be solved in exactly $2$ steps by the above method, resulting in the exact weighted makespan outlined in~\cref{eq:sat-hardness-makespan}, i.e., $29m+9n$ for $n$ variables and $m$ clauses.
	\end{claimproof}
	\begin{claim}
		There does not exist a schedule for $\mathcal{I}_\varphi$ with weighted makespan $29m+9n$ if the Boolean formula $\varphi$ is not satisfiable.
	\end{claim}
	\begin{claimproof}
		It is easy to observe that there is no schedule that solves a variable gadget for $x_i$ in makespan less than ${9(\delta(x_i)+1)}$:
		Building a bridge structure horizontally along the variable gadget takes a quadratic number of moves in the variable gadget's width, and shorter bridges fail to shortcut the north-south distance that the robot would otherwise cover while carrying the tile.
		Furthermore, we can confine tiles to their respective variable gadgets by spacing them appropriately.

		Even if $\varphi$ is not satisfiable, each variable gadget can clearly be solved in a makespan of ${9(\delta(x_i)+1)}$.
		However, this no longer guarantees the existence of a cycle in every clause gadget at some point in time, leaving us to solve at least one clause gadget in a more expensive manner.
		In particular, we incur an extra cost of at least $4$ units per unsatisfied clause, as a cycle containing the cyan tile in \cref{fig:hardness-teleportation-clause} must be created in each clause gadget before the tile can be safely picked up.

		Assuming that $m'\in (1,m]$ is the minimum number of clauses that remain unsatisfied in $\varphi$ for any assignment of $x_i$, we obtain a total makespan of
		\begin{equation}
			\label{eq:hardness-teleportation-greedy}
			2m+4m'+\sum_{i=1}^{n}9(\delta(x_i)+1)
			=4m'+29m+9n
			> 29m+9n.\nonumber
		\end{equation}
		We conclude that a weighted makespan of $29m+9n$ is not achievable and must be exceeded by at least four units if the Boolean formula $\varphi$ is not satisfiable.
	\end{claimproof}
	This concludes our proof for $\lambda=0$.
\end{proof}

\cref{thm:weighted-hardness} is then simply the union of \cref{lem:hardness-lambda-not-zero,lem:hardness-lambda-zero}.

    \section{Constant-factor approximation for 2-scaled instances}\label{sec:bounded-approx}

We now turn to the case in which the configurations have \emph{disjoint bounding boxes}, i.e., there exists an axis-parallel bisector that separates the configurations.
Without loss of generality, let this bisector be horizontal such that the target configuration lies south.
We~present a constant-factor approximation algorithm.

For the remainder of this section, we additionally impose the constraint that both the start and target configurations are \emph{$2$\nobreakdash-scaled}, i.e., they consist of $2 \times 2$-squares of tiles
aligned with a corresponding grid.
In \cref{sec:bounded-stretch}, we extend our result to non-scaled~configurations.

\begin{restatable}{theorem}{thmReconfigTwoScaled}
    \label{thm:reconfig-2-scaled}
    For any $\lambda \in [0, 1]$, there exists a constant $c$
    such that for any pair of $2$-scaled configurations $C_s,C_t\in \configs(n)$ with disjoint bounding boxes, we can efficiently compute a schedule for $C_s\rightrightarrows C_t$ with weighted makespan at most $c\cdot \OPT$.
\end{restatable}

Our algorithm utilizes \emph{histograms} as intermediate configurations.
A~histogram consists of a \emph{base} strip of unit height (a single tile, also when 2-scaled) and unit width \emph{columns} attached orthogonally to its base.
The direction of its columns determines the orientation of a histogram, e.g., $H_s$ in \cref{fig:hist-to-hist} is \emph{north-facing}.

\begin{figure}[htb]
	\centering%
	\includegraphics[page=1]{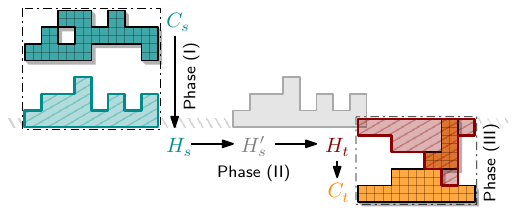}%
	\caption{
		An example for a start and target configuration $C_s, C_t$, the intermediate histograms $H_s, H_t$ with a shared baseline, and the horizontally translated $H_s'$ that shares a tile with $H_t$.
		If $H_s$ and $H_t$ overlap, $H_s' \coloneqq H_s$.
	}
	\label{fig:hist-to-hist}
\end{figure}

As illustrated in~\cref{fig:hist-to-hist}, we proceed in three phases.
\subparagraph*{Phase (I).}%
Iteratively move components of $C_s$ south until it forms a north-facing histogram $H_s$ such that its base shares its $y$-coordinate with a module of $C_t$.
\subparagraph*{Phase (II).}%
Translate $H_s$ to overlap with the target bounding box and transform it into a specific south-facing histogram $H_t$ within the bounding box~of~$C_t$.
\subparagraph*{Phase (III).}%
Finally, apply {\sffamily Phase (I)} in reverse to obtain $C_t$ from $H_t$.
\bigskip

We reduce this to two subroutines:
Transforming any $2$\nobreakdash-scaled configuration into a histogram and reconfiguring any two histograms into one another.
Note that while the respectively obtained histograms are $2$\nobreakdash-scaled when the initial configurations are, the reconfiguration between two opposite-facing histograms actually does not depend on this condition.
In fact, the same reconfiguration routines remain valid even when the histograms involved are not $2$\nobreakdash-scaled, highlighting the broader applicability of our method.

\subsection{Preliminaries for the algorithm}
\label{subsec:preliminaries-algorithm}
For our algorithm, we use the following terms.
Given two configurations $C_s, C_t \in \configs(n)$, the weighted bipartite graph $\bipGraph{C_s, C_t}=({C_s \cup C_t}, {C_s \times C_t}, L_1)$ assigns each edge a weight equal to the $L_1$-distance between its end points.

A \emph{perfect matching} $M$ in $\bipGraph{C_s, C_t}$ is a subset of edges such that every vertex is incident to exactly one of them;
its weight $\weight{M}$ is defined as the sum of its edge weights.
By definition, there exists at least one such matching in $\bipGraph{C_s, C_t}$.
Furthermore, a \emph{minimum weight perfect matching}~(MWPM) is a perfect matching $M$ of minimum weight $\lowerboundOf{C_s,C_t}=\weight{M}$, which is a natural lower bound on the necessary carry distance, i.e., \OPT.

Let $S$ be any schedule for ${C_s\rightrightarrows C_t}$.
Then $S$ induces a perfect matching in~$\bipGraph{C_s, C_t}$, as it moves every tile of~$C_s$ to a distinct position of $C_t$.
We say that~$S$ has \emph{optimal carry distance} exactly if $\cDist{S}=\lowerboundOf{C_s,C_t}$.

We further make the following, simple observation about \emph{crossing paths}, i.e., paths that share a vertex in the integer grid.
This observation will help us to make quick statements concerning MWPMs both for proving \cref{lem:turn-to-histo,lem:hist-to-hist}.
We denote by $L_1(s,t)$ the length of a shortest path between the vertices $s$ and $t$.
\begin{observation}
    \label{obs:swapping}
    If any shortest paths from $s_1$ to $t_1$ and $s_2$ to $t_2$ (all in $\mathbb{Z}^2$) cross,
    then {$L_1(s_1,t_2) + L_1(s_2,t_1) \leq L_1(s_1,t_1) + L_1(s_2,t_2)$} holds.
\end{observation}
This follows from splitting both paths at the crossing vertex and applying triangle~inequality.

\subsection{Phase (I): Transforming into a histogram}
\label{subsec:building-a-histogram}

We proceed by constructing a histogram from an arbitrary $2$-scaled configuration by moving tiles strictly in one cardinal direction.

\begin{lemma}
    \label{lem:turn-to-histo}
    Let $C_s \in \configs(n)$ be a $2$-scaled polyomino and let $H_s$ be a histogram that can be created from $C_s$ by moving tiles in only one cardinal direction.
    We can efficiently compute a schedule with optimal carry distance and total makespan ${\BigO(n + \lowerboundOf{C_s,H_s})}$ for $C_s \rightrightarrows H_s$.
\end{lemma}

To achieve this, we iteratively move sets of tiles in the target direction by two units, until the histogram is constructed.
We give a high-level explanation of our approach by example of a north-facing histogram, as depicted in~\cref{fig:move-down}.

\begin{figure*}[htb]
    \subcaptionsetup{justification=centering}%
    \begin{subfigure}{\columnwidth/4}%
        \centering%
        \includegraphics[page=1]{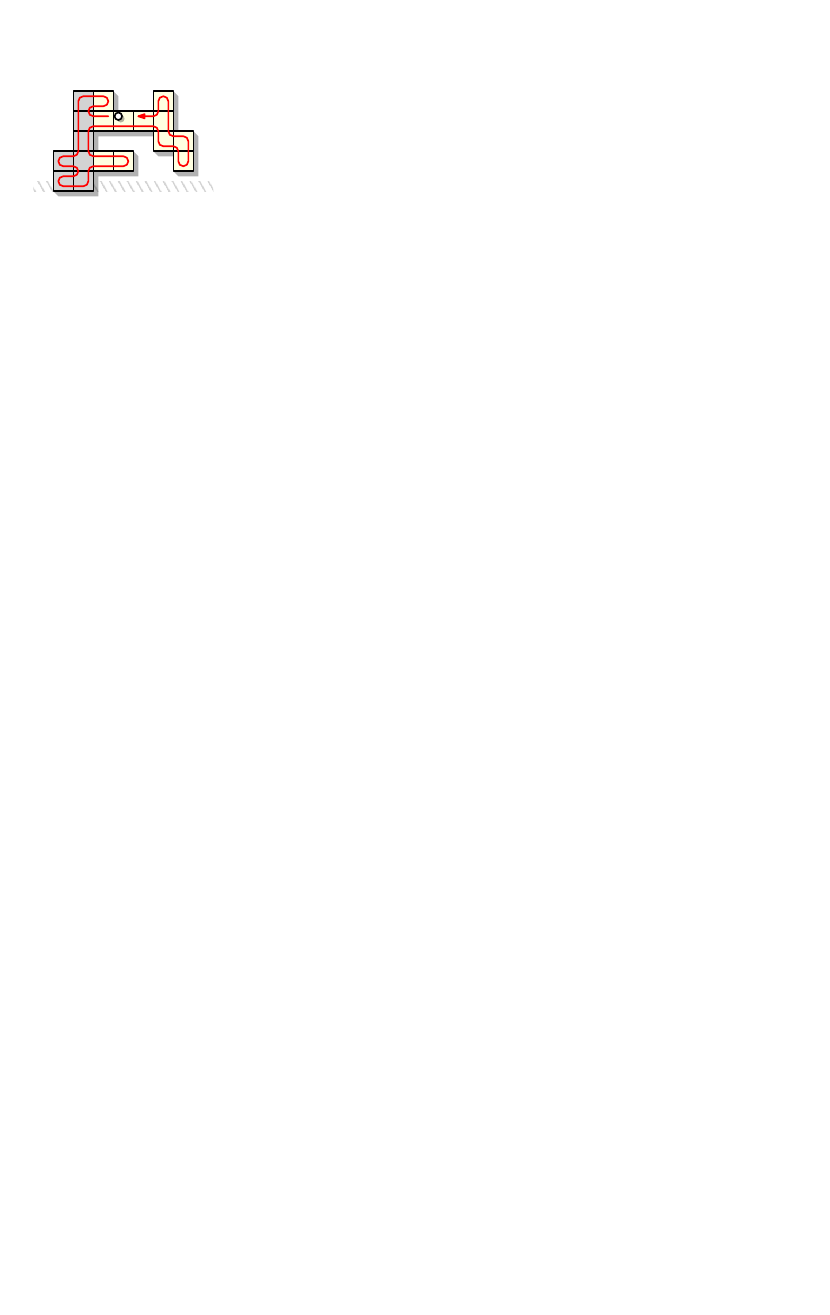}%
        \caption{}%
    \end{subfigure}%
    \begin{subfigure}{\columnwidth*3/4}%
        \centering%
        \includegraphics[page=2]{create-histogram-shortened}%
        \caption{}%
    \end{subfigure}%
    \caption{%
        Visualizing \cref{lem:turn-to-histo}.
        A walk across all tiles \textbf{\textsf{(a)}}, the set $H$ (gray), and two free components~(yellow).
        In \textbf{\textsf{(b)}}, based on the walk, free components are iteratively moved south to reach a histogram shape.
        The free component that is going to be translated next is highlighted~in~green.
    }
    \label{fig:move-down}
\end{figure*}
Let $P$ be any intermediate $2$-scaled polyomino obtained by moving tiles south  while realizing ${C_s \rightrightarrows H_s}$.
Let $H$ be the set of maximal vertical strips of tiles that contain a base tile in $H_s$, i.e., all tiles that do not need to be moved further~south.
We define the \emph{free components} of $P$ as the set of connected components in $P \setminus H$.
By definition, these exist exactly if $P$ is not equivalent to~$H_s$, and once a tile in $P$ becomes part of $H$, it is not moved again until $H_s$ is obtained.

We compute a \emph{walk} that covers the polyomino, i.e., a path that is allowed to traverse edges multiple times and visits each vertex at least once, by depth-first traversal on an arbitrary spanning tree of $P$.
The robot then continuously follows this walk, iteratively moving free components south and updating its path accordingly:
Whenever it enters a free component~$F$ of $P$, it performs a subroutine with makespan $\BigO(\norm{F})$ that moves $F$ south by two units.
Adjusting the walk afterward may increase its length by $\BigO(\norm{F})$ units per \mbox{free component}.

Upon completion of this algorithm, we can bound both the total time spent performing the subroutine and the additional cost incurred by extending the walk by $\BigO(\lowerboundOf{C_s,H_s})$.
Taking into account the initial length of the walk, the resulting makespan is $\BigO(n+\lowerboundOf{C_s,H_s})$.
The subroutine for the translation of free components can be stated as follows.

\begin{lemma}
    \label{lem:move-down}
    We can efficiently compute a schedule of makespan $\BigO(\norm{F})$ to translate any free component~$F$ of a $2$-scaled polyomino in the target direction by two units.
\end{lemma}

\begin{proof}
    Without loss of generality, let the target direction be south.
    We show how to translate~$F$ south by one unit without losing connectivity, which we do twice.

    We follow a bounded-length walk across $F$
    that exclusively visits tiles with a tile neighbor in northern direction.
    Such a walk can be computed by depth-first traversal of $F$.
    Whenever the robot enters a maximal vertical strip of $F$ for the first time, it picks up the northernmost tile, places it at the first unoccupied position to its south, and continues its traversal.

    As each strip is handled exactly once, the total movement cost on vertical strips for carrying tiles and returning to the pre-pickup position is bounded by $\BigO(\norm{F})$.
    This bound also holds for the length of the walk.
\end{proof}

By applying \cref{lem:move-down} on the whole polyomino $P$ instead of just a free component, we can translate $P$
in any direction with asymptotically optimal makespan.

\begin{corollary}
    \label{cor:translation}
    Any $2$-scaled polyomino can be translated by $k$ units in any cardinal direction by a schedule of weighted makespan $\BigO(n\cdot k)$.
\end{corollary}

With this, we are equipped to prove the main result for the first phase.

\begin{proof}[Proof of~\cref{lem:turn-to-histo}]
    As in the proof of \cref{lem:move-down}, assume that the target direction is~south.

    We start by identifying a walk $W$ on $C_s$ that passes every edge at most twice, by depth-first traversal of a spanning tree of $C_s$.
    To construct our schedule, we now simply move the robot along~$W$.
    On this walk, we denote the $i$th encountered free component with $F_i$ as follows.
    Whenever the robot enters a position~$t$ that belongs to a free component~$F_i$, we move~$F_i$ south by two positions (\cref{lem:move-down}) and modify both~$t$ and~$W$ in accordance with the south movement that was just performed.
    The robot stops on the new position for~$t$ and we denote the next free component by~$F_{i+1}$.
    We keep applying \cref{lem:move-down} until~$t$ does not belong to a free component, from where on the robot continues on the modified walk~$W$.
    We illustrate this approach in \cref{fig:move-down}.

    While it is clear that repeatedly moving all free components south eventually creates~$H_s$, it remains to show that this happens within $\BigO(n + \lowerboundOf{C_s,H_s})$ moves.
    To this end, we set $\sumcomps{} = \sum \norm{F_i}$.
    We first show that the schedule has a makespan in $\BigO(n + \sumcomps{})$ and then argue that $\sumcomps{} \in \Theta(\lowerboundOf{C_s, H_s})$.
    Clearly, all applications of \cref{lem:move-down} result in a makespan of~$\BigO(\sumcomps{})$, and since $\norm{W} \in \BigO(n)$ we can also account for the number of steps on the initial walk~$W$.
    However, additional steps may be required:
    Whenever only \emph{parts} of $W$ are moved south by two positions, the distance between the involved consecutive positions in $W$ increases by two, see \cref{fig:walk-enlargement}.
    We account for these additional movements as follows.

    \begin{figure}[htb]
        \centering%
        \includegraphics[page=2]{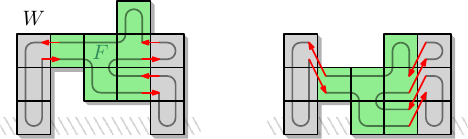}%
        \caption{%
            When all tiles in a free component~$F$ are moved south, the walk~$W$ on the polyomino may become longer.
        }
        \label{fig:walk-enlargement}
    \end{figure}

    For any free component~$F_i$, there are at most $2\norm{F_i}$ edges between $F_i$ and the rest of the configuration.
    Each transfer increases the movement cost by two and each edge is traversed at most twice, increasing the length of $W$ by at most $8\norm{F_i}$ movements for every $F_i$.
    Consequently, $\BigO(n + \sumcomps{})$ also accounts for these movements.

    It remains to argue that $\sumcomps{} \in \Theta(\lowerboundOf{C_s, H_s})$.
    Clearly, as the schedule for moving $F_i$ two units south induces a matching with a weight in $\Theta(\norm{F_i})$, combining all these schedules induces a matching with a weight in $\Theta(\sumcomps{})$.
    We next notice that all schedules for $C_s\rightrightarrows H_s$ that move tiles exclusively south have the same weight because the sum of edge weights in the induced matchings is merely the difference between the sums over all $y$\nobreakdash-positions in~$C_s$ and~$H_s$, respectively.
    Finally, all these matchings have minimal weight, i.e., $\lowerboundOf{C_s, H_s}$:
    If~there exists an MWPM between $C_s$ and $H_s$ with crossing paths, these crossings can be removed by \cref{obs:swapping} without increasing the matching's weight.

    By that, $\sumcomps{} \in \Theta(\lowerboundOf{C_s, H_s})$ holds, yielding a bound of $\BigO(n + \lowerboundOf{C_s,H_s})$.
    Additionally, since tiles are moved exclusively toward the target direction in \cref{lem:move-down}, the schedule has optimal carry distance.
\end{proof}
\subsection{Phase (II): Reconfiguring histograms}
\label{subsec:histogram-to-histogram}
\label{sec:hist-to-hist}

It remains to show how to transform one histogram into the other.
By the assumption of the existence of a horizontal bisector between the bounding boxes of~$C_s$ and~$C_t$, the histogram $H_s$ is north-facing, whereas $H_t$ is south-facing.
The bounding box of $C_s$ is vertically extended to share exactly one $y$-coordinate with the bounding box of $C_t$, and this is where both histogram bases are placed; see~\cref{fig:hist-to-hist} for an illustration.
Note that the histograms may not yet overlap.
However, by \cref{cor:translation}, the tiles in $H_s$ can be moved toward $H_t$ with asymptotically optimal makespan until both histogram bases share a tile.

\begin{lemma}
    \label{lem:hist-to-hist}
    Let $H_s$ and $H_t$ be opposite-facing histograms that share a base tile.
    We can efficiently compute a schedule of makespan $\BigO(n + \lowerboundOf{H_s,H_t})$ for ${H_s \rightrightarrows H_t}$ with optimal carry~distance.
\end{lemma}
\begin{proof}
    We first describe how the schedule $S$ can be constructed and
    then argue that the matching induced by $S$ is an MWPM in~$\bipGraph{H_s, H_t}$.
    We conclude the proof by showing that the schedule is of size $\BigO(n + \lowerboundOf{H_s,H_t})$.

    We write $B_s$ and $B_t$ for the set of all base positions in $H_s$ and $H_t$, respectively, and denote their union by $B \coloneqq B_s \cup B_t$.
    We assume that all robot movements between positions in $H_s$ and $H_t$ are realized along a path that moves vertically until $B$ is reached, continues horizontally on $B$ and then moves vertically to the target position.
    These are shortest paths by construction.

    We denote the westernmost and easternmost position in $B$ by $b_W$ and $b_E$, respectively.
    Without loss of generality, we assume $b_E \in B_t$, because this condition can always be reached by either mirroring the instance horizontally or swapping $B_s$ and $B_t$ and applying the resulting schedule in reverse.

    Assume $b_W \in B_s$.
    We order all tiles in both $H_s \setminus H_t$ and $H_t \setminus H_s$ from west to east and north to south, see \cref{fig:hist-to-hist-greedy-a}.
    Based on this ordering, $S$ is created by iteratively moving the first remaining tile in $H_s \setminus H_t$ to the first remaining position in $H_t \setminus H_s$.
    This way, all moves can be realized on shortest paths over the bases.

    If $b_W \in B_t \setminus B_s$, this is not possible as the westernmost position in $H_t$ may not yet be reachable.
    Therefore, we iteratively extend $B_s$ to the west until the western part of $B_t$ is constructed.
    The tiles are taken from $H_s$ in the same order as before, see~\cref{fig:hist-to-hist-greedy-b}.
    Once we have moved a tile to $b_W$, we can continue as above with $b_W \in B_s$.

	\begin{figure}[htb]
        \subcaptionsetup{justification=centering}%
		\begin{subfigure}[t]{0.5\columnwidth}
			\centering%
			\includegraphics[page=1]{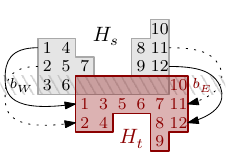}%
			\caption{}
			\label{fig:hist-to-hist-greedy-a}
		\end{subfigure}%
		\begin{subfigure}[t]{0.5\columnwidth}
			\centering%
			\includegraphics[page=2]{hist-to-hist-greedy-more-annotations}%
			\caption{}
			\label{fig:hist-to-hist-greedy-b}
		\end{subfigure}%
		\caption{%
            Position orderings for $H_s \rightrightarrows H_t$.
            Tiles are moved according to the ordering, as indicated by the arrows.
            \textbf{\textsf{(a)}}~The main case of $b_W \in B_s$.
            \textbf{\textsf{(b)}}~Additional steps are required if $b_W \in B_t \setminus B_s$.
        }
		\label{fig:hist-to-hist-greedy}
	\end{figure}

    Let $M$ be the matching induced by $S$.
    We now give an iterative argument why $M$ has the same sum of distances as an arbitrary MWPM $\minMatching$ with $\weight{\minMatching} = \lowerboundOf{H_s,H_t}$.
    First consider $b_W \in B_s$.
    Given the ordering from above, we write $s_i$ and $t_i$ for the $i$th tile position in $H_s \setminus H_t$ and $H_t \setminus H_s$, respectively.
    If $M \neq \minMatching$, there exists a minimal index~$i$ such that $(s_i, t_i) \in M \setminus \minMatching$.
    In $\minMatching$, $s_i$ and $t_i$ are instead matched to other positions $s_j$ and~$t_k$ with~$j, k > i$, i.e., $(s_i, t_k),(s_j, t_i)\in \minMatching$.
    Since $j, k > i$, the positions $s_j$ and $t_k$ do not lie to the west of $s_i$ and $t_i$, respectively.
    Now note that the paths for $(s_i, t_k)$ and $(s_j, t_i)$ always cross:
    Let $b_s$ and $b_t$ be the positions in $B$ that have same $x$-coordinates as $s_i$ and $t_i$, respectively.
    If $s_i$ lies to the east of $t_i$ (case~i), the path from $s_j$ to $t_i$ crosses~$b_s$, and so does any path that starts in $s_i$.
    If $s_i$ does not lie to the east of $t_i$ (case~ii), the path from $s_i$ to $t_k$ crosses $b_t$, and so does any path that ends in $t_i$.
    We illustrate both cases in \cref{fig:hist-to-hist-crossing}.

    \begin{figure}[htb]
        \subcaptionsetup{justification=centering}%
        \begin{subfigure}[t]{0.5\columnwidth}
            \centering%
            \includegraphics[page=1]{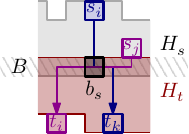}%
            \caption{}%
        \end{subfigure}%
        \begin{subfigure}[t]{0.5\columnwidth}
            \centering%
            \includegraphics[page=2]{hist-to-hist-crossing.pdf}%
            \caption{}%
        \end{subfigure}%
        \caption{%
            The paths for $(s_i, t_k)$ and $(s_j, t_i)$ cross in $b_s$ or $b_t$.
            In \textbf{\textsf{(a)}}, $s_i$ lies to the east of $t_i$ (case~i).
            In \textbf{\textsf{(b)}}, $s_i$ does not lie to the east of $t_i$ (case~ii).
        }
        \label{fig:hist-to-hist-crossing}
    \end{figure}

    As the paths cross, replacing $(s_i, t_k), (s_j, t_i)$ with $(s_i, t_i), (s_j, t_k)$ in $\minMatching$ creates a new matching $\minMatching'$ with $\weight{\minMatching'} \leq \weight{\minMatching}$ due to \cref{obs:swapping}.
    By repeatedly applying this argument, we turn $\minMatching$ into $M$ and note that $\weight{M} \leq \weight{\minMatching}$ holds.
    If $b_W \in B_t \setminus B_s$, case~(i) applies to all tiles moved to extend $B_s$ in western direction, and we can construct $M$ from $\minMatching$ analogously.
    Thus, the matching induced by $S$ is an MWPM.

    Since the carry distance $\cDist{S}$ is exactly $\lowerboundOf{H_s,H_t}$, it remains to bound the empty distance $\eDist{S}$ with $\BigO(n + \lowerboundOf{H_s,H_t})$.
    To this end, we combine all paths from one drop-off position back to its pickup position (in sum $\lowerboundOf{H_s,H_t}$) with all paths between consecutive pickup locations (in sum $\BigO(n)$ because the locations are ordered).
    Since the robot does not actually return to the pickup location, these paths of total length $\BigO(n + \lowerboundOf{H_s,H_t})$ are an upper bound for $\eDist{S}$, and we get $\norm{S} \in \BigO(n + \lowerboundOf{H_s,H_t})$.
\end{proof}

\subsection{Correctness of the algorithm}
\label{sec:makespan}
In the previous sections, we presented schedules for each phase of the overall algorithm.
We~will now leverage these insights to prove the main result of this section, restated here.

\thmReconfigTwoScaled*
\begin{proof}
    By \cref{lem:turn-to-histo,lem:hist-to-hist}, the makespan of the three phases is bounded by, respectively,
    \begin{equation*}
        \BigO(n + \lowerboundOf{C_s,H_s}),\, \BigO(n + \lowerboundOf{H_s,H_t}),\, \BigO(n + \lowerboundOf{H_t,C_t}),
    \end{equation*}
    which we now bound by $\BigO(\lowerboundOf{C_s,C_t})$, proving asymptotic optimality for $C_s \rightrightarrows C_t$.

    Clearly, $n \in \BigO(\lowerboundOf{C_s,C_t})$, as each of the~$n$ tiles has to be moved due to $C_s\cap C_t=\varnothing$.
    We prove a tight bound on the remaining terms, i.e.,
    \begin{equation}
        \label{eq:lower_bound_sum}
        \lowerboundOf{C_s,H_s} + \lowerboundOf{H_s,H_t} + \lowerboundOf{H_t,C_t} = \lowerboundOf{C_s,C_t}.
    \end{equation}

    In {\sffamily Phase (I)}, tiles are moved exclusively toward the bounding box of~$C_t$ along shortest paths to obtain~$H_s$; therefore, $\lowerboundOf{C_s,C_t} = \lowerboundOf{C_s,H_s} + \lowerboundOf{H_s,C_t}$.
    The same applies to the reverse of {\sffamily Phase (III)}, i.e., ${\lowerboundOf{C_t,H_s} = \lowerboundOf{C_t,H_t} + \lowerboundOf{H_t,H_s}}$.
    As the lower bound is symmetric, \cref{eq:lower_bound_sum} immediately follows.
    The total makespan is thus ${\BigO(\lowerboundOf{C_s,C_t})=\BigO(\OPT)}$.
\end{proof}

\Cref{lem:turn-to-histo,lem:hist-to-hist} move tiles on shortest paths and establish schedules that minimize the carry distance.
\Cref{eq:lower_bound_sum} ensures that the combined paths remain shortest possible with regard to $C_s\rightrightarrows C_t$.
Thus, the provided schedule has optimal carry distance.
In~particular, we obtain an optimal schedule when $\lambda = 0$, which corresponds to the case where the robot incurs no cost for movement when not carrying~a~tile.

\begin{corollary}
    An optimal schedule for~${C_s\rightrightarrows C_t}$ can be computed efficiently for any two $2$-scaled configurations ${C_s,C_t\in \configs(n)}$ with disjoint bounding boxes and ${\lambda=0}$.
\end{corollary}

    \section{Constant-factor approximation for general instances}\label{sec:bounded-stretch}

The key advantage of $2$-scaled instances is the absence of cut vertices, which simplifies the maintenance of connectivity during reconfiguration.
Therefore, the challenge with general instances lies in managing cut vertices.

Most parts of our previous method already work independent of the configuration scale.
The only modification required concerns \cref{lem:move-down}, as the polyomino may become disconnected while moving free components that are not $2$-scaled.
To preserve local connectivity during the reconfiguration, we utilize two auxiliary tiles as a patching mechanism; this technique is also employed in other models~\cite{akitaya2021universal,michail2019transformation}.

A key idea is a partitioning of the configuration into horizontal and vertical strips.
With this, our approach are as follows:
To simplify the arguments, we first assume that the robot is capable of holding up to two auxiliary tiles and show that \cref{lem:move-down} still holds for instances that are not $2$-scaled under this assumption.
By using auxiliary tiles, we are able to preserve connectivity while translating horizontal strips; see~\cref{fig:move-down-corridor}.
To~guarantee connectivity, the strips must be moved in a specific order; we resolve this constraint using a recursive strategy that systematically moves dependent strips in the correct sequence.

\begin{figure}[htb]
	\centering%
	\includegraphics[page=1]{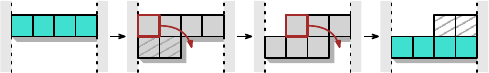}%
	\caption{Using two auxiliary tiles to retain connectivity while moving a unit-height strip south.}
	\label{fig:move-down-corridor}
\end{figure}

\begin{lemma}
	\label{lem:one-scaled-move-down}
	Let the robot hold two auxiliary tiles.
	Given a free component~$F$ on a polyomino~$P$,
	we can efficiently compute a schedule of makespan $\BigO(\norm{F})$ to translate~$F$ in the target direction by one unit.
\end{lemma}

\begin{proof}
	Without loss of generality, we translate $F$ in southern direction.
	We~decompose $F$ into maximal vertical \emph{strips} of unit width.
	Strips that only consist of a single tile are grouped with adjacent single tiles (if they exist) and become maximal horizontal \emph{corridors} of unit height.
	This yields a decomposition of the free component $F$ into strips and corridors, that we collectively refer to as the \emph{elements} $V$ of $F$, see~\cref{fig:one-scaled-decomposition-b}.
	In this construction, each corridor is adjacent to up to two strips and each strip has a height of at least~2.
	All adjacency is either to the west or east side, but never vertical.
	Note that because~$F$ is connected, the graph $G = (V, E)$ of adjacent elements is connected as well.

	\begin{figure}[htb]
		\subcaptionsetup{justification=centering}%
		\begin{subfigure}[t]{\columnwidth/3}
			\centering%
			\includegraphics[page=1]{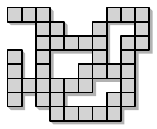}%
			\caption{}%
			\label{fig:one-scaled-decomposition-a}%
		\end{subfigure}%
		\begin{subfigure}[t]{\columnwidth/3}
			\centering%
			\includegraphics[page=2]{decomposition-one-scaled-short-0.85}%
			\caption{}%
			\label{fig:one-scaled-decomposition-b}%
		\end{subfigure}%
		\begin{subfigure}[t]{\columnwidth/3}
			\centering%
			\includegraphics[page=3]{decomposition-one-scaled-short-0.85}%
			\caption{}%
			\label{fig:one-scaled-decomposition-c}%
		\end{subfigure}%
		\caption{%
			A decomposition of a polyomino \textbf{\textsf{(a)}} into gray vertical strips and cyan horizontal corridors \textbf{\textsf{(b)}}.
			From that, we \textbf{\textsf{(c)}} create the adjacency graph $G$ and (by leaving out the dotted edges) a spanning tree $T$.
		}
		\label{fig:one-scaled-decomposition}
	\end{figure}

	We now compute a spanning tree $T$ from~$G$.
	For the remainder of the proof, adjacency of elements is considered to be in $T$.

	For two adjacent elements $A, B \in V$, we say that $A$ \emph{blocks} $B$, if translating all tiles in~$B$ south by one unit would result in all tiles of $A$ losing adjacency to $B$.
	An example for an element $A$ blocking and element $B$ is illustrated in \cref{fig:one-scaled-decomposition-b}: if we translate~$B$ south, we disconnect $A$ from the polyomino.
	Note that $A$ blocks $B$ exactly if the northernmost tile in~$B$ has the same height as the southernmost tile in $A$.
	As at least one of $A$ or $B$ is a strip (with height at least~2), $A$ and $B$ cannot block one another simultaneously.
	Because we translate every tile south by only a single unit, we do not need to update~$T$,
	which means that if $A$ blocks $B$, then once $A$ has been moved south one step, $B$ can be moved south as well without~$A$ and~$B$ losing connectivity.
	If $A$ and $B$ are not blocking one another, they can be moved south in any order.

	We next look at how to move an element $A$ south in the case that $A$ is not blocked by any adjacent element.

	We use two auxiliary tiles which the robot holds at the beginning and end of the operation.
	If $A$ is a strip or a single-tile corridor, we place one tile south of $A$ and then remove the northernmost tile.
	If $A$ is a corridor of two or more tiles, we place the two auxiliary tiles right below the two westernmost tiles of $A$.
	We then move each tile by two steps to the east and one position south, starting from the east.
	The last two tiles become new auxiliary tiles; see \cref{fig:move-down-corridor}.
	Clearly, moving $A$ south has makespan $\BigO(\norm{A})$.

	Finally, the robot has to traverse $T$ such that every element is moved south only after its adjacent blocking elements.
	This can be done recursively, starting at any element in $T$ in the following way.
	At the current element $A$:
	
	\begin{enumerate}
		\item Recurse on all adjacent elements that block $A$ and are not moved south yet.
		\item Move $A$ south.
		\item Recurse on the remaining unvisited elements adjacent to $A$.
	\end{enumerate}

	Because $T$ is connected, every element is visited with this strategy.
	By ordering the adjacent elements, each recursion on an element $A$ requires only $\BigO(\norm{A})$ additional movements to move toward the corresponding elements.

	It remains to show that no element is recursed on more than once.
	Assume that some element $B$ is recursed on twice.
	As there are no cycles in $T$ and all unvisited or blocking neighbors are recursed on in each invocation, $B$ can only be recursed on for a second time by some element $A$ which was previously recursed on by $B$.
	This second invocation stems from the first step of handling $A$ because the third step excludes visited elements, which means that $B$ blocks $A$, and $B$ is not moved south.
	Thus, the first invocation on $B$ did not yet reach the second step.
	In other words, the recursive call on $A$ resulted from $A$ blocking $B$.
	This is a contradiction, as $A$ and $B$ cannot mutually block one another.
\end{proof}

We are now ready to establish our main result for general instances.
For this, it suffices to demonstrate how the auxiliary tiles can be obtained (any two tiles that do not break connectivity can serve this purpose) and that carrying two tiles can be emulated by sequentially carrying one tile after the other.

\begin{theorem}
	\label{thm:reconfig-general}
	For any $\lambda \in [0, 1]$, there exists a constant $c$ such that for any two configurations ${C_s,C_t\in \configs(n)}$ with disjoint bounding boxes, we can efficiently compute a schedule for $C_s\rightrightarrows C_t$ with weighted makespan at most $c\cdot \OPT$.
\end{theorem}

\begin{proof}
	The idea is to execute \cref{lem:one-scaled-move-down} instead of \cref{lem:move-down} every time it is required in \cref{thm:reconfig-2-scaled}.
	To this end, we explain the required modifications for \cref{lem:turn-to-histo}; the same adaptations also work for $H_s \rightrightarrows H_s'$.
	Before we translate the first free component $F$, we move two additional tiles to $F$.
	These tiles can be any tiles from $C_s$ that are not necessary for connectivity (e.g., any leaf from a spanning tree of the dual graph of $C_s$).
	Obtaining these auxiliary tiles requires $\BigO(n)$ movements.
	Next, we sequentially move both tiles from one free component to the next whenever we continue our walk on the configuration.
	We can bound these steps with the length of the walk, thus requiring~$\BigO(n)$ movements in total.

	Now, we only have to argue that \cref{lem:one-scaled-move-down} still holds if the robot initially carries a single tile and must not carry more than one tile at a time:
	Whenever the robot finishes translating the current element, it only picks up one of the two auxiliary tiles and moves it to the next unvisited element.
	This tile is placed in an adjacent position to the target direction of this element (which must be free).
	The robot then moves back and picks up the other auxiliary tile.
	As this extra movement requires~$\BigO(\norm{F})$ additional moves in total, the makespan of \cref{lem:one-scaled-move-down} is still in~$\BigO(\norm{F})$.

	With these adaptations, we can apply \cref{lem:one-scaled-move-down} instead of \cref{lem:move-down} in \cref{thm:reconfig-2-scaled}.
	Moreover, the makespan of \cref{lem:turn-to-histo} is still in $\BigO(n + \lowerboundOf{C_s,H_s})$.
	By that, the analysis of the overall makespan is identical to \cref{thm:reconfig-2-scaled}.
	The only difference is that schedules due to \cref{lem:one-scaled-move-down} do not necessarily have optimal carry distance, which also applies to the overall~schedule.
\end{proof}

    \section{Discussion}\label{sec:discussion}

We presented progress on a reconfiguration problem for two-dimensional tile-based structures (i.e., polyominoes) within abstract material-robot systems.
In particular, we showed that the problem is \NP-hard for any weighting between moving with or without carrying a tile.

Complementary to this negative result, we developed an algorithm to reconfigure two polyominoes into one another in the case that both configurations are contained in disjoint bounding boxes.
The computed schedules are within a constant factor of the optimal reconfiguration schedule.
It is easy to see that our approach can also be used to construct a polyomino, rather than reconfigure one into another.
Instead of deconstructing a start configuration to generate building material, we can assume that tiles are located within a ``depot'' from which they can be picked up.
Performing the second half of the algorithm works as before and builds the target configuration out of tiles from the depot.

Several open questions remain.
A natural open problem is to adapt the approach to instances in which the bounding boxes of the configurations intersect, i.e., they overlap, or are nested.
This hinges on proving a good lower bound on the makespan of any such schedule.
Note that a minimum-weight perfect matching does not provide a reliable lower bound in this setting, as the cost of an optimal solution can be arbitrarily larger; particularly in cases involving small matchings, as visualized in~\cref{fig:rotate-c-matching}.
\begin{figure}[htb]
    \centering
    \includegraphics{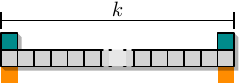}
    \caption{
        Although the MWPM has weight~$4$, the robot is required to perform at least $k - 1 \gg 4$ empty moves.
        This discrepancy shows that the MWPM significantly underestimates the true cost, making it unsuitable as a lower bound for instances with intersecting bounding boxes.
    }
    \label{fig:rotate-c-matching}
\end{figure}

However, an alternative would be to provide an algorithm that achieves worst-case optimality.
Furthermore, the more general question about three dimensional settings remain open.
Our methods could likely be generalized for parallel execution by multiple robots.
More intricate is the question on whether a fully distributed approach is possible.

    \bibliographystyle{plainurl}
    \bibliography{bibliography}

\begin{thebibliography}{10}

\bibitem{AbelAKKS24}
Zachary Abel, Hugo~A. Akitaya, Scott~Duke Kominers, Matias Korman, and Frederick Stock.
\newblock A universal in-place reconfiguration algorithm for sliding cube-shaped robots in a quadratic number of moves.
\newblock In {\em Symposium on Computational Geometry (SoCG)}, pages 1:1--1:14, 2024.
\newblock \href {https://doi.org/10.4230/LIPICS.SOCG.2024.1} {\path{doi:10.4230/LIPICS.SOCG.2024.1}}.

\bibitem{akitaya2021universal}
Hugo~A. Akitaya, Esther~M. Arkin, Mirela Damian, Erik~D. Demaine, Vida Dujmovic, Robin~Y. Flatland, Matias Korman, Bel{\'{e}}n Palop, Irene Parada, Andr{\'{e}} van Renssen, and Vera Sacrist{\'{a}}n.
\newblock Universal reconfiguration of facet-connected modular robots by pivots: The ${O}(1)$ musketeers.
\newblock {\em Algorithmica}, 83:1316--1351, 2021.
\newblock \href {https://doi.org/10.1007/s00453-020-00784-6} {\path{doi:10.1007/s00453-020-00784-6}}.

\bibitem{AkitayaDKKPSSUW22}
Hugo~A. Akitaya, Erik~D. Demaine, Matias Korman, Irina Kostitsyna, Irene Parada, Willem Sonke, Bettina Speckmann, Ryuhei Uehara, and Jules Wulms.
\newblock Compacting squares: Input-sensitive in-place reconfiguration of sliding squares.
\newblock In {\em Scandinavian Symposium on Algorithm Theory (SWAT)}, pages 4:1--4:19, 2022.
\newblock \href {https://doi.org/10.4230/LIPICS.SWAT.2022.4} {\path{doi:10.4230/LIPICS.SWAT.2022.4}}.

\bibitem{parallel-sliding-squares}
Hugo~A. Akitaya, Sándor~P. Fekete, Peter Kramer, Saba Molaei, Christian Rieck, Frederick Stock, and Tobias Wallner.
\newblock Sliding squares in parallel.
\newblock In {\em European Symposium on Algorithms (ESA)}, 2025.
\newblock \href {https://doi.org/10.48550/arXiv.2412.05523} {\path{doi:10.48550/arXiv.2412.05523}}.

\bibitem{almethen2020pushing}
Abdullah Almethen, Othon Michail, and Igor Potapov.
\newblock Pushing lines helps: Efficient universal centralised transformations for programmable matter.
\newblock {\em Theoretical Computer Science}, 830:43--59, 2020.
\newblock \href {https://doi.org/10.1016/j.tcs.2020.04.026} {\path{doi:10.1016/j.tcs.2020.04.026}}.

\bibitem{almethen2022efficient}
Abdullah Almethen, Othon Michail, and Igor Potapov.
\newblock On efficient connectivity-preserving transformations in a grid.
\newblock {\em Theoretical Computer Science}, 898:132--148, 2022.
\newblock \href {https://doi.org/10.1016/j.tcs.2022.09.016} {\path{doi:10.1016/j.tcs.2022.09.016}}.

\bibitem{bfh+-cgur-17}
Aaron~T. Becker, S{\'a}ndor~P. Fekete, Phillip Keldenich, Dominik Krupke, Christian Rieck, Christian Scheffer, and Arne Schmidt.
\newblock Tilt assembly: Algorithms for micro-factories that build objects with uniform external forces.
\newblock {\em Algorithmica}, 82(2):165--187, 2020.
\newblock \href {https://doi.org/10.1007/S00453-018-0483-9} {\path{doi:10.1007/S00453-018-0483-9}}.

\bibitem{2020-Targeted_ICRA}
Aaron~T. Becker, Sándor~P. Fekete, Li~Huang, Phillip Keldenich, Linda Kleist, Dominik Krupke, Christian Rieck, and Arne Schmidt.
\newblock Targeted drug delivery: Algorithmic methods for collecting a swarm of particles with uniform, external forces.
\newblock In {\em International Conference on Robotics and Automation (ICRA)}, pages 2508--2514, 2020.
\newblock \href {https://doi.org/10.1109/ICRA40945.2020.9196551} {\path{doi:10.1109/ICRA40945.2020.9196551}}.

\bibitem{BENLARBI20213598}
Mohamed~Khalil Ben-Larbi, Kattia {Flores Pozo}, Tom Haylok, Mirue Choi, Benjamin Grzesik, Andreas Haas, Dominik Krupke, Harald Konstanski, Volker Schaus, Sándor~P. Fekete, Christian Schurig, and Enrico Stoll.
\newblock Towards the automated operations of large distributed satellite systems. {P}art 1: Review and paradigm shifts.
\newblock {\em Advances in Space Research}, 67(11):3598--3619, 2021.
\newblock \href {https://doi.org/10.1016/j.asr.2020.08.009} {\path{doi:10.1016/j.asr.2020.08.009}}.

\bibitem{9836082}
Christopher~G. Cameron, Zach Fredin, and Neil Gershenfeld.
\newblock Discrete assembly of unmanned aerial systems.
\newblock In {\em International Conference on Unmanned Aircraft Systems (ICUAS)}, pages 339--344, 2022.
\newblock \href {https://doi.org/10.1109/ICUAS54217.2022.9836082} {\path{doi:10.1109/ICUAS54217.2022.9836082}}.

\bibitem{cheung2025assembly}
Kenneth Cheung, Irina Kostitsyna, and Tom Peters.
\newblock Assembly order planning for modular structures by autonomous multi-robot systems.
\newblock In {\em International Conference on Robotics and Automation (ICRA)}, 2025.

\bibitem{Cheung1219}
Kenneth~C. Cheung and Neil Gershenfeld.
\newblock Reversibly assembled cellular composite materials.
\newblock {\em Science}, 341(6151):1219--1221, 2013.
\newblock \href {https://doi.org/10.1126/science.1240889} {\path{doi:10.1126/science.1240889}}.

\bibitem{connor2025transformation}
Matthew Connor and Othon Michail.
\newblock Transformation of modular robots by rotation: $3+1$ musketeers for all orthogonally convex shapes.
\newblock {\em Journal of Computer and System Sciences}, 150:103618, 2025.
\newblock \href {https://doi.org/10.1016/j.jcss.2024.103618} {\path{doi:10.1016/j.jcss.2024.103618}}.

\bibitem{daymude2019computing}
Joshua~J. Daymude, Kristian Hinnenthal, Andr{\'e}a~W. Richa, and Christian Scheideler.
\newblock Computing by programmable particles.
\newblock {\em Distributed Computing by Mobile Entities: Current Research in Moving and Computing}, pages 615--681, 2019.
\newblock \href {https://doi.org/10.1007/978-3-030-11072-7_22} {\path{doi:10.1007/978-3-030-11072-7_22}}.

\bibitem{dbk-obspp-10}
Mark de~Berg and Amirali Khosravi.
\newblock Optimal binary space partitions for segments in the plane.
\newblock {\em International Journal on Computational Geometry and Applications}, 22(3):187--206, 2012.
\newblock \href {https://doi.org/10.1142/S0218195912500045} {\path{doi:10.1142/S0218195912500045}}.

\bibitem{2020-Coordinating_IWOCA}
S{\'a}ndor~P. Fekete.
\newblock Coordinating swarms of objects at extreme dimensions.
\newblock In {\em International Workshop on Combinatorial Algorithms (IWOCA)}, pages 3--13, 2020.
\newblock \href {https://doi.org/10.1007/978-3-030-48966-3_1} {\path{doi:10.1007/978-3-030-48966-3_1}}.

\bibitem{FeketeKKRS23-journal-connected}
S{\'{a}}ndor~P. Fekete, Phillip Keldenich, Ramin Kosfeld, Christian Rieck, and Christian Scheffer.
\newblock Connected coordinated motion planning with bounded stretch.
\newblock {\em Autonomous Agents and Multi-Agent Systems}, 37(2), 2023.
\newblock \href {https://doi.org/10.1007/S10458-023-09626-5} {\path{doi:10.1007/S10458-023-09626-5}}.

\bibitem{FeketeKRS022-journal-labeled-connected}
S{\'{a}}ndor~P. Fekete, Peter Kramer, Christian Rieck, Christian Scheffer, and Arne Schmidt.
\newblock Efficiently reconfiguring a connected swarm of labeled robots.
\newblock {\em Autonomous Agents and Multi-Agent Systems}, 38(2), 2024.
\newblock \href {https://doi.org/10.1007/s10458-024-09668-3} {\path{doi:10.1007/s10458-024-09668-3}}.

\bibitem{fekete2022connected}
S{\'a}ndor~P. Fekete, Eike Niehs, Christian Scheffer, and Arne Schmidt.
\newblock Connected reconfiguration of lattice-based cellular structures by finite-memory robots.
\newblock {\em Algorithmica}, 84(10):2954--2986, 2022.
\newblock \href {https://doi.org/10.1007/s00453-022-00995-z} {\path{doi:10.1007/s00453-022-00995-z}}.

\bibitem{FitchBR03}
Robert Fitch, Zack~J. Butler, and Daniela Rus.
\newblock Reconfiguration planning for heterogeneous self-reconfiguring robots.
\newblock In {\em International Conference on Intelligent Robots and Systems (IROS)}, pages 2460--2467, 2003.
\newblock \href {https://doi.org/10.1109/IROS.2003.1249239} {\path{doi:10.1109/IROS.2003.1249239}}.

\bibitem{FitchBR05}
Robert Fitch, Zack~J. Butler, and Daniela Rus.
\newblock Reconfiguration planning among obstacles for heterogeneous self-reconfiguring robots.
\newblock In {\em International Conference on Robotics and Automation (ICRA)}, pages 117--124, 2005.
\newblock \href {https://doi.org/10.1109/ROBOT.2005.1570106} {\path{doi:10.1109/ROBOT.2005.1570106}}.

\bibitem{friemel2025reconfiguration}
Jonas Friemel, David Liedtke, and Christian Scheffer.
\newblock Efficient shape reconfiguration by hybrid programmable matter.
\newblock In {\em European Workshop on Computational Geometry ({EuroCG})}, pages 14:1--14:8, 2025.
\newblock \href {https://arxiv.org/abs/2501.08663} {\path{arXiv:2501.08663}}.

\bibitem{single-bille-reconfig-IROS}
Javier Garcia, Michael Yannuzzi, Peter Kramer, Christian Rieck, and Aaron~T. Becker.
\newblock Connected reconfiguration of polyominoes amid obstacles using {RRT}$^*$.
\newblock In {\em International Conference on Intelligent Robots and Systems (IROS)}, pages 6554--6560, 2022.
\newblock \href {https://doi.org/10.1109/IROS47612.2022.9981184} {\path{doi:10.1109/IROS47612.2022.9981184}}.

\bibitem{cooperative-bille-reconfig-ICRA}
Javier Garcia, Michael Yannuzzi, Peter Kramer, Christian Rieck, Sándor~P. Fekete, and Aaron~T. Becker.
\newblock Reconfiguration of a {2D} structure using spatio-temporal planning and load transferring.
\newblock In {\em International Conference on Robotics and Automation (ICRA)}, pages 8735--8741, 2024.
\newblock \href {https://doi.org/10.1109/ICRA57147.2024.10611057} {\path{doi:10.1109/ICRA57147.2024.10611057}}.

\bibitem{gmyr2018recognition}
Robert Gmyr, Kristian Hinnenthal, Irina Kostitsyna, Fabian Kuhn, Dorian Rudolph, and Christian Scheideler.
\newblock Shape recognition by a finite automaton robot.
\newblock In {\em International Symposium on Mathematical Foundations of Computer Science ({MFCS})}, pages 52:1--52:15, 2018.
\newblock \href {https://doi.org/10.4230/LIPIcs.MFCS.2018.52} {\path{doi:10.4230/LIPIcs.MFCS.2018.52}}.

\bibitem{gmyr2020forming}
Robert Gmyr, Kristian Hinnenthal, Irina Kostitsyna, Fabian Kuhn, Dorian Rudolph, Christian Scheideler, and Thim Strothmann.
\newblock Forming tile shapes with simple robots.
\newblock {\em Natural Computing}, 19(2):375--390, 2020.
\newblock \href {https://doi.org/10.1007/s11047-019-09774-2} {\path{doi:10.1007/s11047-019-09774-2}}.

\bibitem{ultralight24}
Christine~E. Gregg, Damiana Catanoso, Olivia Irene~B. Formoso, Irina Kostitsyna, Megan~E. Ochalek, Taiwo~J. Olatunde, In~Won Park, Frank~M. Sebastianelli, Elizabeth~M. Taylor, Greenfield~T. Trinh, and Kenneth~C. Cheung.
\newblock Ultralight, strong, and self-reprogrammable mechanical metamaterials.
\newblock {\em Science Robotics}, 9(86), 2024.
\newblock \href {https://doi.org/10.1126/scirobotics.adi2746} {\path{doi:10.1126/scirobotics.adi2746}}.

\bibitem{gregg2018ultra}
Christine~E. Gregg, Joseph~H. Kim, and Kenneth~C. Cheung.
\newblock Ultra-light and scalable composite lattice materials.
\newblock {\em Advanced Engineering Materials}, 20(9):1800213, 2018.
\newblock \href {https://doi.org/10.1002/adem.201800213} {\path{doi:10.1002/adem.201800213}}.

\bibitem{hinnenthal2024efficient}
Kristian Hinnenthal, David Liedtke, and Christian Scheideler.
\newblock Efficient shape formation by {3D} hybrid programmable matter: An algorithm for low diameter intermediate structures.
\newblock In {\em Symposium on Algorithmic Foundations of Dynamic Networks ({SAND})}, pages 15:1--15:20, 2024.
\newblock \href {https://doi.org/10.4230/LIPICS.SAND.2024.15} {\path{doi:10.4230/LIPICS.SAND.2024.15}}.

\bibitem{ItaiPS82}
Alon Itai, Christos~H. Papadimitriou, and Jayme~Luiz Szwarcfiter.
\newblock Hamilton paths in grid graphs.
\newblock {\em {SIAM} Journal on Computing}, 11(4):676--686, 1982.
\newblock \href {https://doi.org/10.1137/0211056} {\path{doi:10.1137/0211056}}.

\bibitem{jenett2017bille}
Ben Jenett and Kenneth Cheung.
\newblock {BILL-E}: Robotic platform for locomotion and manipulation of lightweight space structures.
\newblock In {\em Adaptive Structures Conference (ASC)}, 2017.
\newblock \href {https://doi.org/10.2514/6.2017-1876} {\path{doi:10.2514/6.2017-1876}}.

\bibitem{jenett2019material}
Benjamin Jenett, Amira Abdel-Rahman, Kenneth Cheung, and Neil Gershenfeld.
\newblock Material--robot system for assembly of discrete cellular structures.
\newblock {\em IEEE Robotics and Automation Letters}, 4(4):4019--4026, 2019.
\newblock \href {https://doi.org/10.1109/LRA.2019.2930486} {\path{doi:10.1109/LRA.2019.2930486}}.

\bibitem{jenett2016meso}
Benjamin Jenett, Daniel Cellucci, Christine Gregg, and Kenneth Cheung.
\newblock Meso-scale digital materials: modular, reconfigurable, lattice-based structures.
\newblock In {\em International Manufacturing Science and Engineering Conference (MSEC)}, 2016.
\newblock \href {https://doi.org/10.1115/MSEC2016-8767} {\path{doi:10.1115/MSEC2016-8767}}.

\bibitem{jenett2017design}
Benjamin Jenett, Christine Gregg, Daniel Cellucci, and Kenneth Cheung.
\newblock Design of multifunctional hierarchical space structures.
\newblock In {\em Aerospace Conference}, pages 1--10, 2017.
\newblock \href {https://doi.org/10.1109/AERO.2017.7943913} {\path{doi:10.1109/AERO.2017.7943913}}.

\bibitem{2022-gather_IROS}
Matthias Konitzny, Yitong Lu, Julien Leclerc, S{\'a}ndor~P. Fekete, and Aaron~T. Becker.
\newblock Gathering physical particles with a global magnetic field using reinforcement learning.
\newblock In {\em International Conference on Intelligent Robots and Systems (IROS)}, pages 10126--10132, 2022.
\newblock \href {https://doi.org/10.1109/IROS47612.2022.9982256} {\path{doi:10.1109/IROS47612.2022.9982256}}.

\bibitem{KostitsynaOPPSS24}
Irina Kostitsyna, Tim Ophelders, Irene Parada, Tom Peters, Willem Sonke, and Bettina Speckmann.
\newblock Optimal in-place compaction of sliding cubes.
\newblock In {\em Scandinavian Symposium on Algorithm Theory (SWAT)}, pages 31:1--31:14, 2024.
\newblock \href {https://doi.org/10.4230/LIPICS.SWAT.2024.31} {\path{doi:10.4230/LIPICS.SWAT.2024.31}}.

\bibitem{michail2019transformation}
Othon Michail, George Skretas, and Paul~G. Spirakis.
\newblock On the transformation capability of feasible mechanisms for programmable matter.
\newblock {\em Journal of Computer and System Sciences}, 102:18--39, 2019.
\newblock \href {https://doi.org/10.1016/j.jcss.2018.12.001} {\path{doi:10.1016/j.jcss.2018.12.001}}.

\bibitem{NiesReconfig}
Eike Niehs, Arne Schmidt, Christian Scheffer, Daniel~E. Biediger, Michael Yannuzzi, Benjamin Jenett, Amira Abdel-Rahman, Kenneth~C. Cheung, Aaron~T. Becker, and Sándor~P. Fekete.
\newblock Recognition and reconfiguration of lattice-based cellular structures by simple robots.
\newblock In {\em International Conference on Robotics and Automation (ICRA)}, pages 8252--8259, 2020.
\newblock \href {https://doi.org/10.1109/ICRA40945.2020.9196700} {\path{doi:10.1109/ICRA40945.2020.9196700}}.

\bibitem{santos2022light}
Ana~L. Santos, Dongdong Liu, Anna~K. Reed, Aaron~M. Wyderka, Alexis van Venrooy, John~T. Li, Victor~D. Li, Mikita Misiura, Olga Samoylova, Jacob~L. Beckham, Ciceron Ayala-Orozco, Anatoly~B. Kolomeisky, Lawrence~B. Alemany, Antonio Oliver, George~P. Tegos, and James~M. Tour.
\newblock Light-activated molecular machines are fast-acting broad-spectrum antibacterials that target the membrane.
\newblock {\em Science Advances}, 8(22):eabm2055, 2022.
\newblock \href {https://doi.org/10.1126/sciadv.abm2055} {\path{doi:10.1126/sciadv.abm2055}}.

\bibitem{zhang2019robotic}
Zhuoran Zhang, Xian Wang, Jun Liu, Changsheng Dai, and Yu~Sun.
\newblock Robotic micromanipulation: Fundamentals and applications.
\newblock {\em Annual Review of Control, Robotics, and Autonomous Systems}, 2:181--203, 2019.
\newblock \href {https://doi.org/10.1146/ANNUREV-CONTROL-053018-023755} {\path{doi:10.1146/ANNUREV-CONTROL-053018-023755}}.

\bibitem{Yangsheng:five_mm_robot}
Yangsheng Zhu, Mingjing Qi, Zhiwei Liu, Jianmei Huang, Dawei Huang, Xiaojun Yan, and Liwei Lin.
\newblock A 5-mm untethered crawling robot via self-excited electrostatic vibration.
\newblock {\em IEEE Transactions on Robotics}, 38(2):719--730, 2022.
\newblock \href {https://doi.org/10.1109/TRO.2021.3088053} {\path{doi:10.1109/TRO.2021.3088053}}.

\end{thebibliography}

\end{document}